\documentclass[conference]{IEEEtran}
\IEEEoverridecommandlockouts

\usepackage{cite}
\usepackage{amsmath,amssymb,amsfonts,amsthm}
\usepackage[linesnumbered,lined]{algorithm2e}
\usepackage{graphicx}
\usepackage{textcomp}
\usepackage{xcolor}
\def\BibTeX{{\rm B\kern-.05em{\sc i\kern-.025em b}\kern-.08em
    T\kern-.1667em\lower.7ex\hbox{E}\kern-.125emX}}
    
\theoremstyle{definition}
\newtheorem{definition}{Definition}[section]

\newtheoremstyle{theorem}%		% Name
	{}%							% Space above
	{}%							% Space below
	{\itshape}%							% Body font
	{}%							% Indent amount
	{\bfseries}%							% Theorem head font
	{.}%							% Punctuation after theorem head
	{ }%							% Space after theorem head, ' ', or \newline
	{}%							% Theorem head spec (can be left empty, meaning `normal')
\theoremstyle{theorem}
\newtheorem{thm}{Theorem}
\newtheorem{prop}[thm]{Proposition}

\newcommand{\coloneq}{\mathrel{{\colon}{=}}}
\DeclareMathOperator{\sgn}{sgn}

\allowdisplaybreaks

\begin{document}

\title{Mapping of Internet ``Coastlines'' via Large Scale Anonymized Network Source Correlations
\thanks{DISTRIBUTION STATEMENT A. Approved for public release. Distribution is unlimited. 
This material is based upon work supported by the Under Secretary of Defense for Research and Engineering under Air Force Contract No. FA8702-15-D-0001. Any opinions, findings, conclusions or recommendations expressed in this material are those of the author(s) and do not necessarily reflect the views of the Under Secretary of Defense for Research and Engineering. Research was also sponsored by the United States Air Force Research Laboratory and the Department of the Air Force Artificial Intelligence Accelerator and was accomplished under Cooperative Agreement Number FA8750-19-2-1000. The views and conclusions contained in this document are those of the authors and should not be interpreted as representing the official policies, either expressed or implied, of the Department of the Air Force or the U.S. Government.
\copyright\ 2023 Massachusetts Institute of Technology. Delivered to the U.S. Government with Unlimited Rights, as defined in DFARS Part 252.227-7013 or 7014 (Feb 2014). Notwithstanding any copyright notice, U.S. Government rights in this work are defined by DFARS 252.227-7013 or DFARS 252.227-7014 as detailed above. Use of this work other than as specifically authorized by the U.S. Government may violate any copyrights that exist in this work.}
}

\author{\IEEEauthorblockN{Hayden Jananthan$^1$, Jeremy Kepner$^1$, Michael Jones$^1$, William Arcand$^1$, David Bestor$^1$, William Bergeron$^1$, \\ Chansup Byun$^1$, Timothy Davis$^2$, Vijay Gadepally$^1$,  Daniel Grant$^3$, Michael Houle$^1$, Matthew Hubbell$^1$, \\  Anna Klein$^1$, Lauren Milechin$^1$, Guillermo Morales$^1$, Andrew Morris$^3$, Julie Mullen$^1$, Ritesh Patel$^1$, \\ Alex Pentland$^1$, Sandeep Pisharody$^1$,  Andrew Prout$^1$,  Albert Reuther$^1$, Antonio Rosa$^1$, Siddharth Samsi$^1$, \\ Tyler Trigg$^1$, Gabriel Wachman$^1$, Charles Yee$^1$, Peter Michaleas$^1$
\\
\IEEEauthorblockA{$^1$MIT, $^2$Texas A\&M, $^3$GreyNoise
}}}

\maketitle

\begin{abstract}
	Expanding the scientific tools available to protect computer networks can be aided by a deeper understanding of the underlying statistical distributions of network traffic and their potential geometric interpretations. Analyses of large scale network observations provide a unique window into studying those underlying statistics. Newly developed GraphBLAS hypersparse matrices and D4M associative array technologies enable the efficient anonymized analysis of network traffic on the scale of trillions of events. This work analyzes over 100,000,000,000 anonymized packets from the largest Internet telescope (CAIDA) and over 10,000,000 anonymized sources from the largest commercial honeyfarm (GreyNoise). Neither CAIDA nor GreyNoise actively emit Internet traffic and provide distinct observations of unsolicited Internet traffic (primarily botnets and scanners). Analysis of these observations confirms the previously observed Cauchy-like distributions describing temporal correlations between Internet sources. The Gull lighthouse problem is a well-known geometric characterization of the standard Cauchy distribution and motivates a potential geometric interpretation for Internet observations. This work generalizes  the Gull lighthouse problem to accommodate larger classes of coastlines, deriving a closed-form solution for the resulting probability distributions, stating and examining the inverse problem of identifying an appropriate coastline given a continuous probability distribution, identifying a geometric heuristic for solving this problem computationally, and applying that heuristic to examine the temporal geometry of different subsets of network observations. Application of this method to the CAIDA and GreyNoise data reveals a several orders of magnitude difference between known benign and other traffic which can lead to potentially novel ways to protect networks.
\end{abstract}

\begin{IEEEkeywords}
heavy-tailed distribution, cauchy distribution, cybersecurity, statistics, internet modeling
\end{IEEEkeywords}

\section{Introduction}

As the role of the Internet in day-to-day activities within modern civilization continually deepens its roots in existing contexts and spreads its roots even further, the necessity for---and complexity within---cybersecurity and cyber-defense is correspondingly made ever more apparent \cite{cisco2022cisco, caida2019anonymized, kepner2022zero, claffy2000measuring, li2023survey, rabinovich2016measuring, claffy2020workshop}. An important set of tools for supporting cyber-defense are Internet measurements. Enabling some of those measurements are internet observatories and outposts such as the Center for Applied Internet Data Analysis (CAIDA) Telescope \cite{caida2023ucsd} and the GreyNoise honeyfarm \cite{greynoise2023greynoise}. The CAIDA Telescope monitors a continuous stream of packets from an unsolicited darkspace making up approximately 1/256 of the Internet. CAIDA Telescope traffic is dominated by likely malicious traffic (e.g., vulnerability scanners and backscatter from denial-of-service attacks). The GreyNoise honeyfarm engages in active conversations with potentially malicious traffic in order to associate various meta-tags to those sources \cite{greynoise2023greynoise}. These network sensors generate trillions of events where ensuring privacy is a paramount consideration.  Newly developed GraphBLAS hypersparse matrices and D4M associative array technologies enable the efficient anonymized analysis of these data at this scale and have revealed a variety of phenomena \cite{kepner16mathematical, buluc17design, kepner2018mathematics, davis2019algorithm}.

Along with packet information and other raw measurements, aggregate measurements can be computed and include the computation of aggregate network-theoretic quantities like source fan-out, destination fan-in, the degree distribution of sources and destinations, and the numbers of packets sent from sources to destinations. These aggregate network-theoretic quantities often follow heavy-tailed probability distributions given by the modified Zipf-Mandelbrot model \cite{jones2022graphblas, kawaminami2022large}. The ubiquity of such heavy-tailed distributions across fields---financial \cite{pareto1964cours, gabaix1999zipf, anderson2006long, hackett196770, cont2001empirical, loretan1994testing}, natural language analysis \cite{estoup1916gammes, zipf1949human}, and network analysis \cite{broder2000graph, delvin2021hybrid, huberman1999growth, crovella1998heavy, mahanti2013tale, faloutsos1999on}---and the fact that many traditional light-tailed statistical techniques like the use of finite-order moments and the Central Limit Theorem fail badly in the heavy-tailed case \cite{nair2022fundamentals, poisson1829suite, poisson1827probabilite, demirandacardoso2021graphical, cherapanamjeri2020algorithms, stehlik2010favorable, benaych-georges2014central} has underlined the importance of exploring the underlying statistical distributions describing the Internet.

\begin{figure*}[htbp]
	\begin{center}
		\includegraphics[width=2\columnwidth]{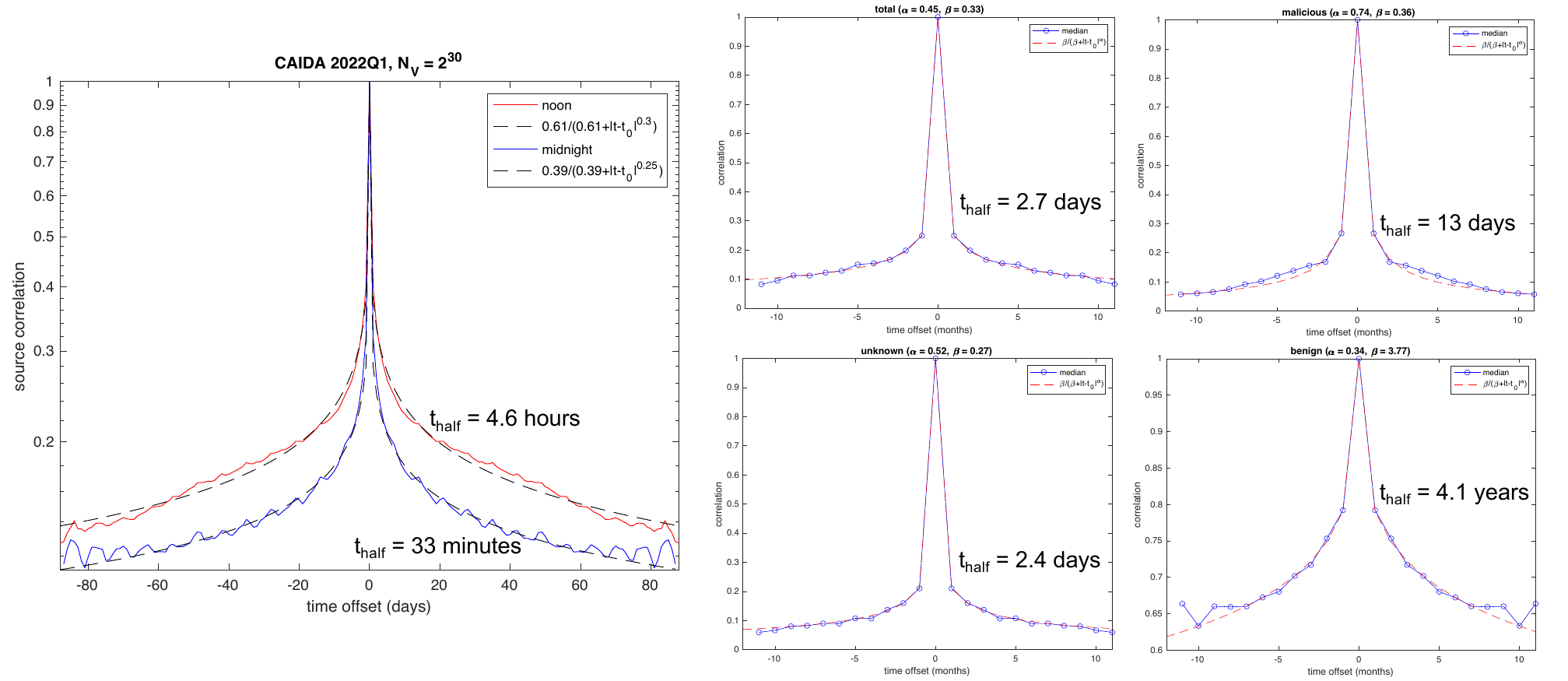}
	\end{center}
	\caption{(left) Source self-correlations among sources observed by the CAIDA darknet telescope during 2022Q1 each point represent the sources drawn a packet window with $N_V = 2^{30}$ valid packets.  (right) Self-correlations among different categories of sources (total, malicious, unknown, and benign) in the GreyNoise honeyfarm from 2021Q2 thru 2022Q1. Corresponding modified Cauchy parameters and full-width-half-maximum time $t_{\rm half} = \beta^{1/\alpha}$ are shown illustrating the significant difference between benign and other traffic.}
\label{self correlation by classification total figure}
\end{figure*}

In additional to localized statistics, long-term temporal correlations can be computed between internet observatories and outposts, such as what fraction of sources from one internet observatory's given collection period are seen by another internet observatory as a function of time before and after the first observatory's initial collection period. Prior works \cite{kepner2022temporal} examine these quantities in detail and find that they may be described via a modified Cauchy distribution of the form
\begin{equation*}
	p(x) \propto \frac{\beta}{\beta + x^\alpha}
\end{equation*}
for some values of $\alpha, \beta$. When $\alpha = 2$, this becomes the well-known Cauchy distribution.  A geometric interpretation for the Cauchy distribution is phrased in terms of lighthouses and coastlines. The empirical similarity between the models found in \cite{kepner2022temporal} and the Cauchy distribution---as well as the potential similarities between the function of lighthouses and the function of internet observatories/outposts---make it conceivable that a similar geometric interpretation could be made for the aforementioned modified Cauchy distributions; shedding light on the underlying geometry of the Internet.

The outline of the remainder of the paper is as follows.  First, \S \ref{modified cauchy distribution section} presents new GraphBLAS and D4M enabled correlation analysis of recent large-scale CAIDA Telescope and GreyNoise honeyfarm observations that affirm the previously observed modified Cauchy distributions. Next, \S \ref{background section} briefly describes the definition and some relevant properties of the Cauchy distribution and presents a derivation of the Cauchy distribution based on the Gull lighthouse problem.  In \S \ref{from coastlines to probability distributions section} we define the probability density function for the distribution of $x$-coordinates of points on a sufficiently nice curve $y = f(x)$ given uniformly random azimuths from a fixed point not on the curve. In \S \ref{from probability distributions to coastlines section} we examine the inverse problem---given a continuous probability distribution determining an appropriate coastline function that gives rise to the original distribution in the manner described in \S \ref{from coastlines to probability distributions section}. Finally, \S \ref{conclusion and future work section} recaps the major contributions and ideas discussed in this paper and future directions. Appendices examine the mathematical derivations and subtleties in more detail.  For the remainder of the paper, we assume all random variables $X$ are continuous random variables supported on some real interval, e.g., $\mathbb{R}$, $[0, \infty)$, etc. Given a random variable $X$, $p_X$ denotes the probability density function, and $\overline{F}_X$ denotes the complementary cumulative distribution function.

\section{Large-Scale Temporal Correlations} \label{modified cauchy distribution section}

\begin{figure*}[htbp]
	\begin{center}
		\includegraphics[width=1.7\columnwidth]{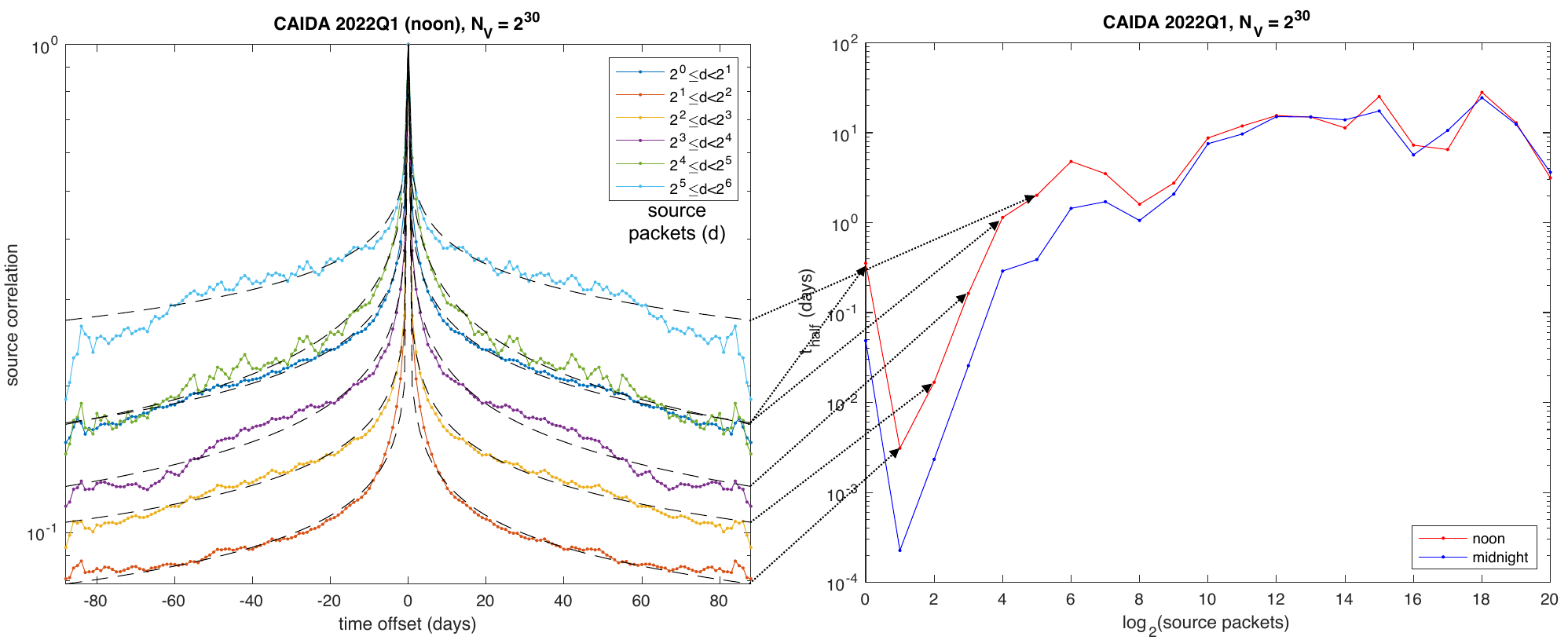}
	\end{center}
	\caption{(left) CAIDA noon self-correlations amongst sources with $2^i \leq d < 2^{i+1}$ packets over the range $0 \leq i \leq 5$ along with corresponding best-fit modified Cauchy distribution. (right) CAIDA noon and midnight full-width-half-maximum time $t_{\rm half}$ over the source packet ranges corresponding to $0 \leq i \leq 20$. The correlation time is typically days, with a notable exception in the range  $2^1 \leq d < 2^2$ which has a correlation time of a few minutes.} \label{self correlation by degree figure}
\end{figure*}

The CAIDA darknet telescope is a significant portion of a globally routed /8 network carrying essentially no legitimate traffic since it is an Internet darkspace, providing an ideal vantage point by which to observe and study unsolicited anomalous traffic. The GreyNoise honeyfarm is made up of thousands of servers carrying out conversations with sources scanning the Internet; based on these conversations GreyNoise can associate various metadata with those sources to collectively build a refined picture of the malicious sources regularly scanning the Internet and the techniques they employ. 

CAIDA collects over 1,000,000,000,000 unique packets each month from hundreds of millions of unique sources, while GreyNoise converses with several millions of unique sources each month. The sheer amount of data requires advanced technology to analyze \cite{lumsdaine2007challenges, kolda2009tensor, hilbert2011world}, making essential use of the MIT SuperCloud supercomputing center as well as massively parallel GraphBLAS and D4M hierarchical hypersparse matrices. The latter technologies have made it possible to analyze hundreds of billions of packets in minutes 
\cite{jones2022graphblas, kepner2009parallel, kepner2011graph, kepner2018mathematics}.
These prior analyses have revealed consistent scaling relations across both the space- and time-domains \cite{jones2022graphblas, kepner2022temporal, kawaminami2022large, delvin2021hybrid, kepner2022new}. 

Among these analyses is previous work examining and modeling the temporal correlations between sources seen by both CAIDA and GreyNoise \cite{kepner2022temporal}.  Specifically, the cross-correlation of CAIDA sources with varying numbers of packets $2^{i} \leq d < 2^{i+1}$  with GreyNoise data over a 15 month span were computed.
The source cross-correlations were fit to Gaussian, Cauchy, and modified Cauchy distributions, with the data being well-approximated by the modified Cauchy distribution with exponent $\alpha \approx 3/4$.

\begin{definition}[modified Cauchy distribution]
	A modified Cauchy distribution is any probability distribution with probability density function
	\begin{equation*}
		p(x) \propto \frac{\beta}{\beta + |x|^\alpha}
	\end{equation*}
	where $\alpha, \beta > 0$, denoted $\mathrm{ModCauchy}(\beta, \alpha)$.
\end{definition}

\noindent When $\alpha \leq 1$, the integral $\int_{-\infty}^\infty{\frac{\beta}{\beta + x^\alpha}dx}$ diverges, necessitating restriction of the distribution to a bounded domain. While being supported on a bounded domain means that the notion of being heavy-tailed does not apply, the fact that the Cauchy distribution is heavy-tailed and ratio of the Cauchy distribution to the modified Cauchy distribution
\begin{equation*}
	\lim_{x \to \pm\infty}{\frac{\frac{\beta}{\beta + |x|^\alpha}}{\frac{1}{1 + x^2}}} = \lim_{x \to \pm\infty}{\frac{\beta (1 + x^2)}{\beta + |x|^\alpha}} = \infty
\end{equation*}
whenever $0 < \alpha < 2$ suggests that the modified Cauchy distribution acts like a heavy-tailed distribution.
The full-width-half-maximum $x_{\rm half}$ of the modified Cauchy distribution is given by
	\begin{equation*}
		\frac{1}{2} = \frac{\beta}{\beta + |x_{\rm half}|^\alpha}
	\end{equation*}
resulting in $x_{\rm half} = \beta^{1/\alpha}$ and is a useful measure of correlation decay rate.

While \cite{kepner2022temporal} examined the temporal correlations between sources seen by two \emph{different} internet observatories, the applicability of the modified Cauchy distributions can be seen to extend to the self-correlations of sources within recent CAIDA and GreyNoise data sets.

Figure~\ref{self correlation by classification total figure} (left) shows self-correlations among anonymized sources observed at noon or midnight by the CAIDA  telescope in 2022Q1.  Each data point represents the sources drawn from a valid packet window with $N_V = 2^{30}$ CAIDA  packets collected over a few minutes.  This sample represents over 100,000,000,000 packets in total.  Such volumes of data are ideal for GraphBLAS hierarchical hypersparse matrices to process in a timely manner \cite{jones2022graphblas}.

Figure~\ref{self correlation by classification total figure} (right) shows self-correlations among different categories of anonymized sources (total, malicious, unknown, and benign) in the GreyNoise honeyfarm from 2021Q2 thru 2022Q1.   The string format and smaller scale of the GreyNoise sources ($\approx$10,000,000) are well-suited for D4M hierarchical associative arrays \cite{jones2022graphblas}.

Both Figure~\ref{self correlation by classification total figure} (left) and (right) show the parameters for the modified Cauchy distributions which closely model the self-correlation plots, producing modified Cauchy distributions with exponents $\alpha_\mathrm{total} \approx 0.45$, $\alpha_\mathrm{noon} \approx 0.3$, $\alpha_\mathrm{midnight} \approx 0.25$, $\alpha_\mathrm{unknown} \approx 0.52$, $\alpha_\mathrm{malicious} \approx 0.74$, and $\alpha_\mathrm{benign} \approx 0.34$.  Similarly shown are the corresponding full-width-half-maximum time $t_{\rm half}$ for each distribution and illustrates the significant difference in correlation time between benign traffic (years) and other traffic (hours to days).

The CAIDA data has sufficient volume to be able to further look at the self-correlations within ranges of packets.  Figure~\ref{self correlation by degree figure} (left) shows the CAIDA self-correlations amongst sources with $2^i \leq d < 2^{i+1}$ packets over the range $0 \leq i \leq 5$ along with corresponding best-fit modified Cauchy distributions, further affirming the model.    Figure~\ref{self correlation by degree figure} (right) shows the corresponding full-width-half-maximum time $t_{\rm half}$ over the source packet ranges corresponding to $0 \leq i \leq 20$, illustrating that the correlation time is typically days, with a notable exception in the range  $2^1 \leq d < 2^2$, which has a correlation time of a few minutes.

\section{Cauchy Distribution---Properties \& Geometry} \label{background section}

The standard Cauchy distribution has a geometric interpretation often phrased in terms of lighthouses and coastlines, which may offer geometric insight into Internet data.  First, recall the definition of the Cauchy distribution.

\begin{definition}[Cauchy distribution \cite{nair2022fundamentals}]
	The Cauchy distribution is the continuous probability distribution with probability density function
	\begin{equation*}
		p(x) = \frac{1}{\pi (1 + x^2)}
	\end{equation*}
\end{definition}

\noindent The Cauchy distribution is often used as a canonical example of a ``pathological'' continuous distribution, having no finite moments of order greater than or equal to $1$ \cite{nair2022fundamentals, poisson1827probabilite, poisson1829suite}. Additionally, the Cauchy distribution is a heavy-tailed distribution.

\begin{definition}[heavy-tailed \cite{nair2022fundamentals}]
	A random variable $X$ is \emph{right heavy-tailed} if its support contains $[M, \infty)$ for some $M \in \mathbb{R}$ and its complementary cumulative distribution function $\overline{F}_X$ satisfies, for any $\mu > 0$,
	\begin{equation*}
		\limsup_{x \to \infty}{\frac{\overline{F}_X(x)}{e^{-\mu x}}} = \infty
	\end{equation*}
	$X$ is \emph{left heavy-tailed} if $-X$ is right heavy-tailed. $X$ is \emph{heavy-tailed} if it is either left or right heavy-tailed. A statistical distribution is \emph{heavy-tailed} if any (equivalently, every) random variable with that distribution is heavy-tailed.
\end{definition}

\begin{prop}[\cite{nair2022fundamentals}]
	The Cauchy distribution is both left and right heavy-tailed.
\end{prop}
\begin{proof}
	We use l'Hopital's Rule to show that $\lim_{x \to \infty}{e^{\mu x}\left(\frac{1}{2} - \frac{1}{\pi}\arctan(x)\right)} = \infty$, based on the observation that $\lim_{x \to \infty}{\left(\frac{1}{2} - \frac{1}{\pi}\arctan(x)\right)} = \lim_{x \to \infty}{e^{-\mu x}} = 0$. Indeed,
	\begin{align*}
		\lim_{x \to \infty}{\frac{\frac{d}{dx}\left(\frac{1}{2} - \frac{1}{\pi}\arctan(x)\right)}{\frac{d}{dx}\left(e^{-\mu x}\right)}} & = \lim_{x \to \infty}{\frac{-\frac{1}{\pi} \frac{1}{1+x^2}}{-\mu e^{-\mu x}}} \\
		& = \lim_{x \to \infty}{\frac{1}{\pi \mu} \frac{e^{\mu x}}{1 + x^2}} \\
		& = \infty
	\end{align*}
	Showing the Cauchy distribution is left heavy-tailed is analogous.
\end{proof}

A traditional geometric interpretation of the Cauchy distribution is by way of the Gull lighthouse problem (see Figure~\ref{GullLighthouseProblem}) and is stated as follows

\begin{quote}
\cite{gull1988bayesian}
A lighthouse is somewhere off a piece of straight coastline at position $x_0$ along the coast and a distance $y$ out to sea. It emits a series of short, highly collimated flashes at random intervals and hence at random azimuths [$\theta$]. These pulses are intercepted on the coast by photo-detectors that record only the fact that a flash has occurred, but not the azimuth from which it came. $N$ flashes have so far been recorded at positions $x_i$ ($i = 1, 2, \ldots, N$). Where is the lighthouse?
\end{quote}

While the original problem is concerned with finding the location of the lighthouse given a finite sequence of recorded positions, the Bayesian approach inverts the problem by first determining what the probability distribution of coastline positions is given the location of the lighthouse has been fixed---this distribution is our Cauchy distribution. Carrying out this analysis, the relation between the azimuth $\theta$ and the corresponding coastline position $x$ is $x = x_0 + y \tan(\theta)$ or, alternatively, $\theta = \arctan\left(\frac{x - x_0}{y}\right)$, as seen in Figure~\ref{cauchy setup figure}. Let $X$ and $\Theta$ be the random variables for coastline position and azimuth, respectively. Making use of the probability transformation rule, we then have
\begin{align*}
	p_X(x) & = p_\Theta(\theta) \left|\frac{d\theta}{dx}\right|
	= \frac{1}{\pi} \frac{y}{y^2 + (x - x_0)^2}
\end{align*}
When $x_0 = 0$ and $y = 1$, we get the familiar Cauchy distribution.  For the remainder of the paper we assume $x_0 = 0$ for simplicity.

\begin{figure}[htbp]
	\begin{center}
		\includegraphics[width=0.85\columnwidth]{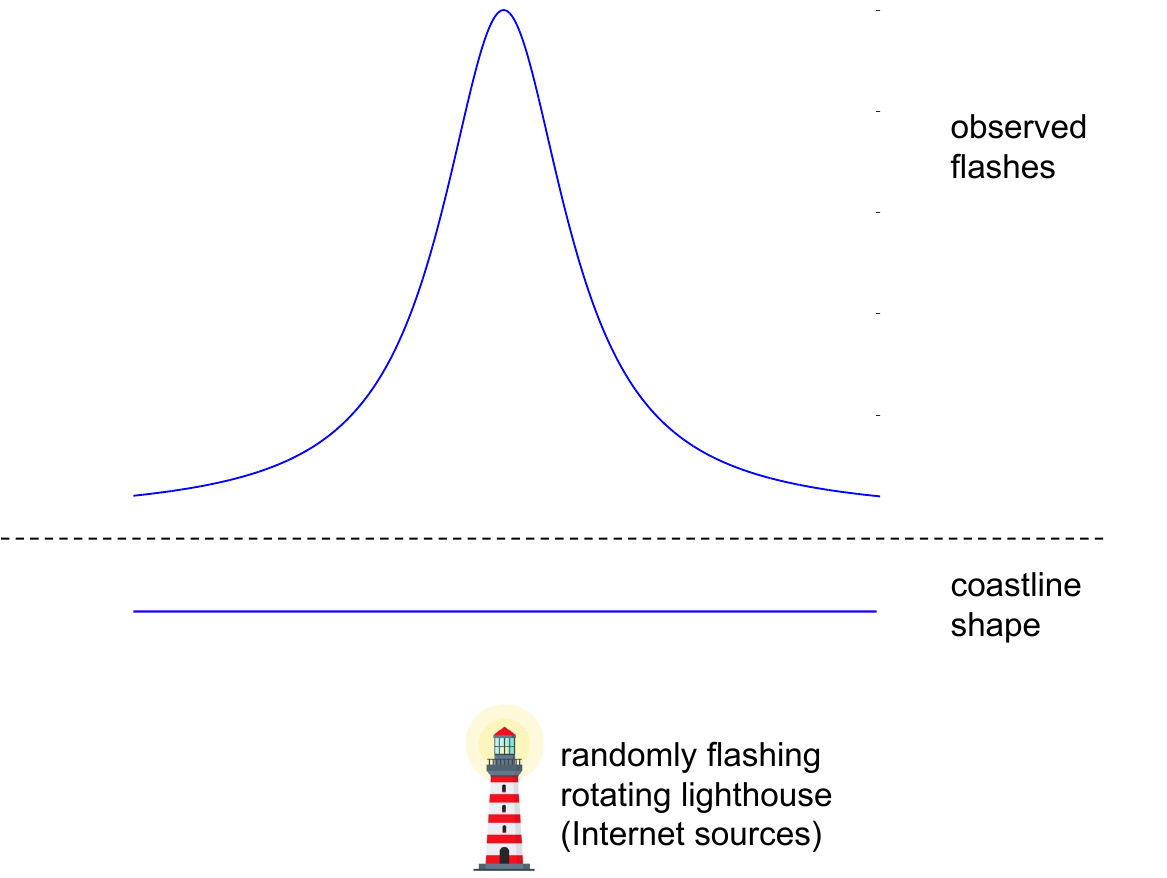}
	\end{center}
	\caption{Gull lighthouse problem predicts the distribution of flashes seen at any point on a coastline by a randomly flashing lighthouse.}
	\label{GullLighthouseProblem}
\end{figure}

\begin{figure}[htbp]
	\begin{center}
		\includegraphics[width=0.65\columnwidth]{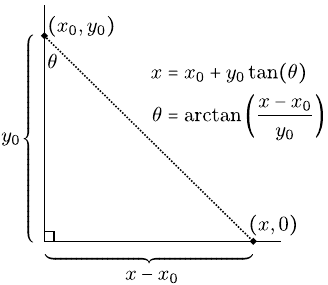}
	\end{center}
	\caption{Geometric setup for a straight coastline ($y = 0$) with respect to a lighthouse located at $(x_0, y_0)$. The azimuth $\theta$ and the corresponding coastline position $(x, 0)$ are related by the equations $x = x_0 + y_0 \tan(\theta)$ and $\theta = \arctan((x - x_0)/y_0)$.}
	\label{cauchy setup figure}
\end{figure}

\section{Coastlines to Probability Distributions} \label{from coastlines to probability distributions section}

A geometric derivation can be carried out similar to that of the lighthouse-coastline interpretation of the Cauchy distribution that generalizes the situation to allow for a much wider class of coastlines. Namely, we require that a real-valued function $f$ of a real variable satisfy the following condition for the corresponding probability distribution to be defined.

\begin{definition}[coastline condition]
	A differentiable function $f \colon (a, b) \to \mathbb{R}$ (where $-\infty \leq a < b \leq \infty$) satisfies the \emph{coastline condition} with respect to a point $(0, y_0) \in \mathbb{R}^2$ if:
	\begin{itemize}
		\item $(0, y_0)$ does not belong to the graph of $f$.
		\item The function $\theta \colon (a, b) \to (-\pi, \pi)$ defined by
		\begin{equation*}
			\theta(x) \coloneq \begin{cases} \arctan\left(\dfrac{x}{y_0 - f(x)}\right) & \text{if $y_0 \geq f(x)$,} \\ \arctan\left(\dfrac{x}{f(x) - y_0}\right) & \text{otherwise,} \\ \quad + \sgn(x) \dfrac{\pi}{2} & {} \end{cases}
		\end{equation*}
		is injective.
	\end{itemize}
		In this case, the limits $\lim_{x \to a^+}{\theta(x)}$ and $\lim_{x \to b^-}{\theta(x)}$ are the \emph{azimuthal bounds of $f$}.\footnote{The limits $\lim_{x \to a^+}{\theta(x)}$ and $\lim_{x \to b^-}{\theta(x)}$ automatically exist by the Monotone Convergence Theorem as $\theta$ is bounded and monotonic ($\theta$ is continuous since $f$ is differentiable so injectivity implies monotonicity).}
\end{definition}

Given a point $(0, y_0) \in \mathbb{R}^2$ as a stand-in for our lighthouse and a coastline $f \colon (a, b) \to \mathbb{R}$ satisfying the coastline condition above, the following distribution has the geometric interpretation that it is the probability distribution for the random variable $X$ which is the $x$-coordinate of a point on the graph of $f$ determined uniquely from an azimuthal angle $\Theta$ chosen uniformly at random from some azimuthal bounds. See Appendix A for more details about the derivation.

\begin{definition}[generalized Cauchy distribution]
	Given a point $(0, y_0) \in \mathbb{R}^2$ and a differentiable function $f \colon (a, b) \to \mathbb{R}$ satisfying the coastline condition with respect to $(0, y_0)$, a random variable is $\mathrm{GenCauchy}(0, y_0; f)$ if 
	\begin{equation} \label{generalized lighthouse distribution}
		p_X(x) = \frac{1}{\beta - \alpha} \frac{(y_0 - f(x)) + x f'(x)}{x^2 + (y_0 - f(x))^2}
	\end{equation}
	where $-\pi \leq \alpha < \beta \leq \pi$ are the azimuthal bounds of $f$.
\end{definition}

As a sanity check we can show that Equation~\ref{generalized lighthouse distribution} is consistent with the motivating example of the Cauchy distribution. Substituting the flat coastline $f(x) \coloneq 0$, azimuthal bounds $-\pi/2 < \theta < \pi/2$, and lighthouse position $(0, y_0) = (0, 1)$ in Equation~\ref{generalized lighthouse distribution} yields $p_X(x) = \frac{1}{\pi} \cdot \frac{1}{1 + x^2}$, as desired.

As an additional example, consider the case of a lower unit semicircle coastline centered at the lighthouse position $(0, y_0) \coloneq (0, 1)$, so that $f \colon [-1, 1] \to \mathbb{R}$ is defined by $f(x) \coloneq 1 - \sqrt{1 - x^2}$ and azimuthal bounds $-\pi/2 \leq \theta \leq \pi/2$. Then 
\begin{align*}
	p_X(x) & = \frac{1}{\pi} \cdot \frac{1 - (1 - \sqrt{1 - x^2}) + x \cdot \frac{-x}{\sqrt{1 - x^2}}}{x^2 + (1 - (1 - \sqrt{1 - x^2}))^2} \\
	& = \frac{1}{\pi} \cdot \left(\sqrt{1 - x^2} - \frac{x^2}{\sqrt{1 - x^2}}\right) \\
	& = \frac{1}{\pi} \cdot \frac{1}{\sqrt{1 - x^2}}
\end{align*}
Another characterization of this distribution $\mathrm{GenCauchy}(0, 1; \text{semicircle})$ is as an arcsine-distributed random variable supported on $[-1, 1]$, $\mathrm{Arcsine}(1, 1)$, and agrees with existing geometric interpretations of arcsine-distributed random variables in terms of the $x$-coordinate of a point uniformly chosen from the circumference of a circle.

\section{Probability Distributions to Coastlines} \label{from probability distributions to coastlines section}

Given a continuous probability distribution with probability density function $p$ supported on the interval $I$, when can we express $p$ in the form
\begin{equation*}
	p(x) \propto \frac{(y_0 - f(x)) + x f'(x)}{x^2 + (f(x) - y_0)^2}
\end{equation*}
for some appropriate coastline $f$ and lighthouse position $(0, y_0)$? One method to do so makes use of the fact that (fixing $(0, y_0) \in \mathbb{R}$ and $\alpha, \beta \in (-\pi, \pi)$) the existence of such an $f$ necessitates the existence of a solution to the first-order differential equation
\begin{equation} \label{main differential equation explicit form}
	y' = \frac{1}{x} \left(y - (\beta - \alpha) p(x) (x^2 + y^2)\right)
\end{equation}
where $y = y_0 - f(x)$. There are several difficulties with this approach:
\begin{itemize}
	\item Equation~\ref{main differential equation explicit form} is a nonexact, nonlinear ODE.
	\item The function 
	\begin{equation*}
		(x, y) \mapsto \frac{1}{x} \left(y - (\beta - \alpha) p(x) (x^2 + y^2)\right)
	\end{equation*}
	has an essential discontinuity on the line $x = 0$, preventing standard existence theorems like the Carath\'{e}odory's Existence Theorem from implying the existence of local solutions on neighborhoods intersecting the line $x = 0$ which is problematic when $0 \in (\alpha, \beta)$. 
	\item Standard techniques like first order finite differences and Runge-Kutta do not seem to address the general case well.
\end{itemize}
Even in the case where local solutions can be shown to exist, to have a complete geometric interpretation of the given probability distribution we must have a \emph{global} solution $y = f(x)$ supported on the same interval $I$ upon which our probability distribution is supported. Moreover, $f$ must be such that the rays from the lighthouse position $(0, y_0)$ corresponding to the azimuths $\alpha$ and $\beta$ intersect the graph of $f$ at the endpoints of $I$. 

Instead, a numerical approach making use of a geometric heuristic is taken. First, some geometric setup: Suppose a coastline $f$ is known along with a lighthouse position $(0, y_0)$ and azimuthal bounds $\alpha \leq \theta \leq \beta$. Given $\theta \in (\alpha, \beta)$, let:
\begin{itemize}
	\item $P = (x, f(x))$ is the point on the coastline $f$ corresponding to the azimuth $\theta$;
	\item $\ell$ is the line tangent to $f$ at $P$;
	\item $\varphi \in [0, \pi)$ is the measure of the angle swept from the ray $\overrightarrow{LP}$ to $\ell$, where $L = (0, y_0)$;
	\item $\omega \in [0, \pi/2]$ is the measure of the acute angle formed between $\ell$ and the $x$-axis when $\ell$ is not horizontal and $\omega = 0$ when $\ell$ is horizontal.
\end{itemize}
Figure~\ref{heuristic setup figure} illustrates this setup in the case where $\ell$ has positive slope and $P$ is in the upper-half-plane. When $\ell$ has positive slope, we may calculate $\omega = \theta + \varphi - \frac{\pi}{2}$ while when $\ell$ has negative slope we instead have $\omega = \frac{\pi}{2} - \theta - \varphi$. When $\ell$ has positive slope then $f'(x) = \tan(\omega)$; when $\ell$ has negative slope then $f'(x) = -\tan(\omega) = \tan(-\omega)$; finally, when $\ell$ has zero slope then $\theta = \frac{\pi}{2} - \varphi$ so $\tan(\theta + \varphi - \frac{\pi}{2}) = \tan(0) = 0 = f'(x)$. In all cases we find that 
\begin{equation*}
	f'(x) = \tan\left(\theta + \varphi - \frac{\pi}{2}\right)
\end{equation*}
meaning that for $0 < \delta$ small
\begin{equation*}
	f(x + \delta) \approx f(x) + \delta \tan\left(\theta + \varphi - \frac{\pi}{2}\right)
\end{equation*}

\begin{figure}[htbp]
	\begin{center}
		\includegraphics[width=0.65\columnwidth]{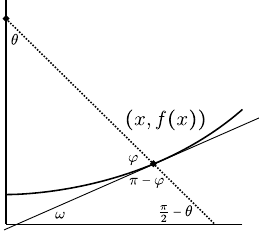}
	\end{center}
	\caption{Illustration of the angle measures identified for use in the sine-squared heuristic $p(x) \propto \sin(\varphi(x))^2$, where $x$ is the first coordinate of the point on the coastline $y = f(x)$ corresponding to the azimuth $\theta$. The remaining angle measures $\varphi, \omega$ are taken with respect to the line tangent to $y = f(x)$ at $(x, f(x))$.}
	\label{heuristic setup figure}
\end{figure}

With this setup in mind, we may approximate $f$ by making the assertion
\begin{equation*}
	p(x) = \beta^{-1} \sin(\varphi(x))^2
\end{equation*}
for some $\beta > 0$, or, since $\varphi \in [0, \pi]$,
\begin{equation} \label{geometric heuristic}
	\varphi(x) = \arcsin\left(\sqrt{\beta p(x)}\right) \tag*{$(\ast)$}
\end{equation}
and iteratively calculate $f(x_i + \delta_i)$ from $f(x_i)$ and step-sizes $\delta_i \coloneq x_{i+1} - x_i$. Algorithm~\ref{heuristic algorithm} computes a sequence $(f_i)_{i = 1}^N$ from a predetermined lighthouse position $(0, y_0)$, a probability density function $p$ supported on $[a, b]$, and a partition $a = x_1 < x_2 < \cdots < x_N = b$ of $[a, b]$; the sequence $(f_i)_{i = 1}^N$ approximates a coastline function $f \colon [a, b] \to \mathbb{R}$ with $f(x_i) \coloneq f_i$ for $i = 1, 2, \ldots, N$. Appendix B
%~\ref{limitations of sine-square heuristic appendix}
examines some of the limitations of this sine-squared heuristic.

\begin{algorithm}
	\SetKwInOut{Input}{input}\SetKwInOut{Output}{output}
	\Input{lighthouse position $(0, y_0)$\; \\ $\beta$ a proportionality constant\; \\ $p$ a probability density function supported on $[a, b]$\; \\ $x = (x_i)_{i = 1}^N \in [a, b]^N$ a strictly increasing sequence with $x_1 = a$ and $x_N = b$\;}
	\Output{$f = (f_i)_{i=1}^N$ approximate coastline $y$-coordinates\;}
	\BlankLine
	initialization\;
	$\delta \leftarrow (x_{i+1} - x_{i})_{i=1}^{N-1}$\;
	$\varphi \leftarrow \left(\arcsin(\sqrt{\beta p(x_{i})})\right)_{i=1}^N$\;
	$f \leftarrow (0)_{i=1}^N$\;
	$\theta \leftarrow (0)_{i=1}^{N-1}$\;
	\For{$i \leftarrow 1$ \KwTo $N-1$}{
		\eIf{$f_i \leq y_0$}{
			$\displaystyle \theta_i \leftarrow \arctan\left(\frac{x_i}{y_0 - f_i}\right)$\;
		}{
			$\displaystyle \theta_i \leftarrow \arctan\left(\frac{f_i - y_0}{x_i}\right) + \sgn(x_i) \frac{\pi}{2}$\;
		}
		$\displaystyle f_{i+1} \leftarrow f_i + \delta_i \tan\left(\theta_i + \varphi_i - \frac{\pi}{2}\right)$\;
	}
	\caption{Approximation of coastline function $f$ from given probability density function $p \colon [a, b] \to \mathbb{R}$ on partition $a = x_1 < x_2 < \cdots < x_N = b$ using the sine-squared heuristic $\varphi(x) = \arcsin(\sqrt{\beta p(x)})$. A lighthouse position $(0, y_0)$ and a proportionality constant $\beta$ must be defined ahead of time and the value $f(0)$ is implicitly set to $0$.} \label{heuristic algorithm}
\end{algorithm}

Applying this sine-squared heuristic to the self-correlations shown in Figure~\ref{self correlation by classification total figure} produces the coastlines seen in Figure~\ref{caida coastlines figure}. This geometric representation of the  CAIDA and GreyNoise source correlations for different CAIDA collection times and GreyNoise classifications illustrate the significant geometric differences between coastlines of these observations. Most apparent is the stark difference between sources classified by GreyNoise as benign  and all other sources.  The 100x separation seen in benign coastlines provides a potentially useful geometric interpretation of this network traffic.

\begin{figure*}[htbp]
	\begin{center}
		\includegraphics[width=1.7\columnwidth]{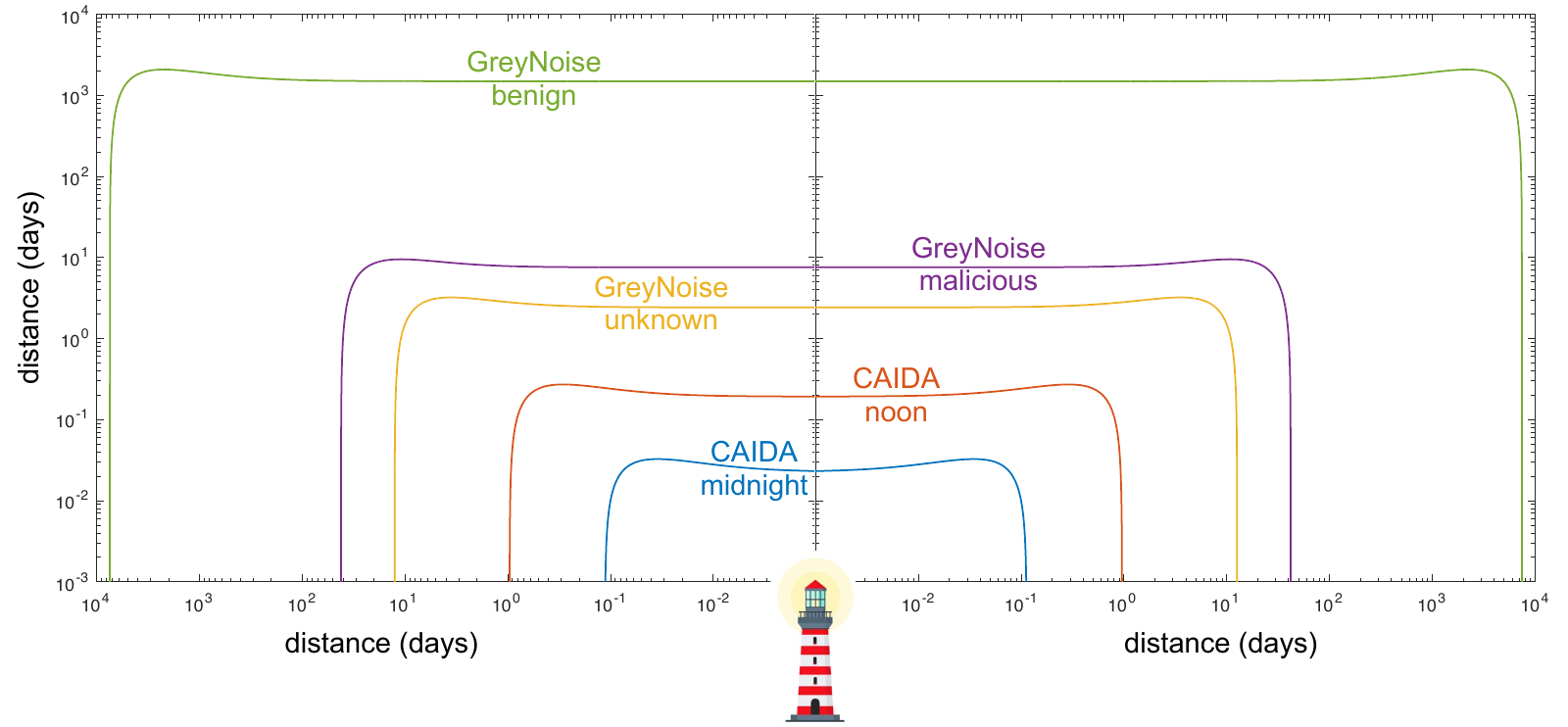}
	\end{center}
	\caption{Coastlines generated for the modified Cauchy distributions associated with CAIDA noon/midnight and GreyNoise malicious/unknown/benign observations. The over 100x separation between GreyNoise benign and the other coastlines provides a potentially useful geometric interpretation of this network traffic.} \label{caida coastlines figure}
\end{figure*}

\section{Conclusion \& Future Work} \label{conclusion and future work section}

Large scale network observations provide a unique window on the Internet.  Newly developed GraphBLAS hypersparse matrices and D4M associative array technologies enable the efficient anonymized analysis of these data at large scale.   This work analyzes over 100,000,000,000 anonymized packets from the largest Internet telescope (CAIDA) and over 10,000,000 anonymized sources from the largest commercial honeyfarm (GreyNoise).   Extending the geometric interpretation of the Cauchy distribution in terms of lighthouses and coastlines provides new geometric interpretations of a much wider class of distributions than seen previously. This paves the road towards further geometric tools for studying fields based on their statistical distributions, with a particular eye towards temporal correlations in the cyber-sphere where it has been consistently observed that Cauchy-like probability distributions (so-called modified Cauchy distributions) are ubiquitous. We have examined the mathematical difficulties inherent in this approach, namely in the direction of finding an appropriate coastline function given a probability density function, and suggested and evaluated a geometric heuristic for numerically solving this problem.  Applying this heuristic to CAIDA and GreyNoise reveals 100x separation of benign coastlines which may be a potentially useful geometric interpretation of this network traffic.

Future work will focus on further refining the sine-squared heuristic, including the identification of properties of probability density functions that might lead to better or worse performance such as first and second derivative signs and determination of best practices addressing those properties. Further examination of the first-order ordinary differential equation Equation~\ref{main differential equation explicit form} is also needed in order to determine whether any existing theory guarantees existence and/or uniqueness of local solutions and under what conditions (if any) local existence and uniqueness can be extended to global existence and uniqueness. Likewise, further exploration of the landscape of numerical techniques for solving ill-behaved ordinary differential equations might lead to alternative numerical methods for solving the aforementioned problem of finding a coastline function given a probability density function.

\section*{Acknowledgment}

The authors wish to acknowledge the following individuals for their contributions and support: Daniel Andersen, Sean Atkins, Chris Birardi, Bob Bond, Andy Bowne, Stephen Buckley, K Claffy, Cary Conrad, Chris Demchak, Alan Edelman, Garry Floyd, Jeff Gottschalk, Dhruv Gupta, Chris Hill, Charles Leiserson, Kirsten Malvey, Chase Milner, Sanjeev Mohindra,  Dave Martinez, Joseph McDonald, Heidi Perry, Christian Prothmann, Steve Rejto, Josh Rountree, Daniela Rus, Mark Sherman, Scott Weed, Adam Wierman, Marc Zissman.

\bibliographystyle{IEEEtran}
\bibliography{cypher_bib}

\begin{thebibliography}{10}
\providecommand{\url}[1]{#1}
\csname url@samestyle\endcsname
\providecommand{\newblock}{\relax}
\providecommand{\bibinfo}[2]{#2}
\providecommand{\BIBentrySTDinterwordspacing}{\spaceskip=0pt\relax}
\providecommand{\BIBentryALTinterwordstretchfactor}{4}
\providecommand{\BIBentryALTinterwordspacing}{\spaceskip=\fontdimen2\font plus
\BIBentryALTinterwordstretchfactor\fontdimen3\font minus
  \fontdimen4\font\relax}
\providecommand{\BIBforeignlanguage}[2]{{%
\expandafter\ifx\csname l@#1\endcsname\relax
\typeout{** WARNING: IEEEtran.bst: No hyphenation pattern has been}%
\typeout{** loaded for the language `#1'. Using the pattern for}%
\typeout{** the default language instead.}%
\else
\language=\csname l@#1\endcsname
\fi
#2}}
\providecommand{\BIBdecl}{\relax}
\BIBdecl

\bibitem{cisco2022cisco}
\BIBentryALTinterwordspacing
``Cisco visual networking index: Forecast and trends, 2018–2023,'' Jan 2022.
  [Online]. Available:
  \url{https://www.cisco.com/c/en/us/solutions/collateral/executive-perspectives/annual-internet-report/white-paper-c11-741490.html}
\BIBentrySTDinterwordspacing

\bibitem{caida2019anonymized}
\BIBentryALTinterwordspacing
``{CAIDA} anonymized internet traces dataset ({A}pril 2008 - {J}anuary 2019).''
  [Online]. Available:
  \url{https://www.caida.org/catalog/datasets/passive\_dataset/}
\BIBentrySTDinterwordspacing

\bibitem{kepner2022zero}
J.~Kepner, J.~Bernays, S.~Buckley, K.~Cho, C.~Conrad, L.~Daigle, K.~Erhardt,
  V.~Gadepally, B.~Greene, M.~Jones, R.~Knake, B.~Maggs, P.~Michaleas,
  C.~Meiners, A.~Morris, A.~Pentland, S.~Pisharody, S.~Powazek, A.~Prout,
  P.~Reiner, K.~Suzuki, K.~Takahashi, T.~Tauber, L.~Walker, and D.~Stetson,
  ``Zero botnets: An observe-pursue-counter approach,'' \emph{Belfer Center
  Reports}, 2022.

\bibitem{claffy2000measuring}
K.~Claffy, ``Measuring the internet,'' \emph{IEEE Internet Computing}, vol.~4,
  no.~1, pp. 73--75, 2000.

\bibitem{li2023survey}
\BIBentryALTinterwordspacing
B.~Li, J.~Springer, G.~Bebis, and M.~{Hadi Gunes}, ``A survey of network flow
  applications,'' \emph{Journal of Network and Computer Applications}, vol.~36,
  no.~2, pp. 567--581, 2013. [Online]. Available:
  \url{https://www.sciencedirect.com/science/article/pii/S1084804512002676}
\BIBentrySTDinterwordspacing

\bibitem{rabinovich2016measuring}
M.~Rabinovich and M.~Allman, ``Measuring the internet,'' \emph{IEEE Internet
  Computing}, vol.~20, no.~4, pp. 6--8, 2016.

\bibitem{claffy2020workshop}
\BIBentryALTinterwordspacing
k.~claffy and D.~Clark, ``Workshop on internet economics ({WIE} 2019) report,''
  \emph{SIGCOMM Comput. Commun. Rev.}, vol.~50, no.~2, p. 53–59, may 2020.
  [Online]. Available: \url{https://doi.org/10.1145/3402413.3402421}
\BIBentrySTDinterwordspacing

\bibitem{caida2023ucsd}
\BIBentryALTinterwordspacing
``{\it {UCSD} Network Telescope},''
  https://www.caida.org/projects/network\_telescope/. [Online]. Available:
  \url{https://www.caida.org/projects/network\_telescope/}
\BIBentrySTDinterwordspacing

\bibitem{greynoise2023greynoise}
\BIBentryALTinterwordspacing
``Greynoise.'' [Online]. Available: \url{https://greynoise.io/}
\BIBentrySTDinterwordspacing

\bibitem{kepner16mathematical}
J.~{Kepner}, P.~{Aaltonen}, D.~{Bader}, A.~{Bulu{\c{c}}}, F.~{Franchetti},
  J.~{Gilbert}, D.~{Hutchison}, M.~{Kumar}, A.~{Lumsdaine}, H.~{Meyerhenke},
  S.~{McMillan}, C.~{Yang}, J.~D. {Owens}, M.~{Zalewski}, T.~{Mattson}, and
  J.~{Moreira}, ``Mathematical foundations of the {G}raph{BLAS},'' in
  \emph{2016 {IEEE} High Performance Extreme Computing Conference ({HPEC})},
  2016, pp. 1--9.

\bibitem{buluc17design}
A.~{Bulu{\c{c}}}, T.~{Mattson}, S.~{McMillan}, J.~{Moreira}, and C.~{Yang},
  ``Design of the {G}raph{BLAS} {API} for {C},'' in \emph{2017 IEEE
  International Parallel and Distributed Processing Symposium Workshops
  (IPDPSW)}, 2017, pp. 643--652.

\bibitem{kepner2018mathematics}
J.~Kepner and H.~Jananthan, \emph{Mathematics of big data: Spreadsheets,
  databases, matrices, and graphs}.\hskip 1em plus 0.5em minus 0.4em\relax MIT
  Press, 2018.

\bibitem{davis2019algorithm}
T.~A. Davis, ``Algorithm 1000: {S}uite{S}parse: {G}raph{BLAS}: Graph algorithms
  in the language of sparse linear algebra,'' \emph{ACM Transactions on
  Mathematical Software (TOMS)}, vol.~45, no.~4, pp. 1--25, 2019.

\bibitem{jones2022graphblas}
M.~Jones, J.~Kepner, D.~Andersen, A.~Buluç, C.~Byun, K.~Claffy, T.~Davis,
  W.~Arcand, J.~Bernays, D.~Bestor, W.~Bergeron, V.~Gadepally, M.~Houle,
  M.~Hubbell, H.~Jananthan, A.~Klein, C.~Meiners, L.~Milechin, J.~Mullen,
  S.~Pisharody, A.~Prout, A.~Reuther, A.~Rosa, S.~Samsi, J.~Sreekanth,
  D.~Stetson, C.~Yee, and P.~Michaleas, ``{G}raph{BLAS} on the edge: Anonymized
  high performance streaming of network traffic,'' in \emph{2022 {IEEE} High
  Performance Extreme Computing Conference ({HPEC})}, Sep. 2022, pp. 1--8.

\bibitem{kawaminami2022large}
I.~Kawaminami, A.~Estrada, Y.~Elsakkary, H.~Jananthan, A.~Buluç, T.~Davis,
  D.~Grant, M.~Jones, C.~Meiners, A.~Morris, S.~Pisharody, and J.~Kepner,
  ``Large scale enrichment and statistical cyber characterization of network
  traffic,'' in \emph{2022 IEEE High Performance Extreme Computing Conference
  (HPEC)}, Sep. 2022, pp. 1--7.

\bibitem{pareto1964cours}
V.~Pareto, \emph{Cours d'{\'e}conomie politique}.\hskip 1em plus 0.5em minus
  0.4em\relax Librairie Droz, 1964, vol.~1.

\bibitem{gabaix1999zipf}
\BIBentryALTinterwordspacing
X.~Gabaix, ``Zipf's law and the growth of cities,'' \emph{The American Economic
  Review}, vol.~89, no.~2, pp. 129--132, 1999. [Online]. Available:
  \url{http://www.jstor.org/stable/117093}
\BIBentrySTDinterwordspacing

\bibitem{anderson2006long}
C.~Anderson, \emph{The long tail: Why the future of business is selling more
  for less}.\hskip 1em plus 0.5em minus 0.4em\relax Hachette Books, 2006.

\bibitem{hackett196770}
A.~P. Hackett, ``70 years of best sellers, 1895-1965,'' \emph{New York: RR
  Bowker Company}, 1967.

\bibitem{cont2001empirical}
\BIBentryALTinterwordspacing
R.~Cont, ``Empirical properties of asset returns: stylized facts and
  statistical issues,'' \emph{Quantitative Finance}, vol.~1, no.~2, pp.
  223--236, 2001. [Online]. Available: \url{https://doi.org/10.1080/713665670}
\BIBentrySTDinterwordspacing

\bibitem{loretan1994testing}
\BIBentryALTinterwordspacing
M.~Loretan and P.~C. Phillips, ``Testing the covariance stationarity of
  heavy-tailed time series: An overview of the theory with applications to
  several financial datasets,'' \emph{Journal of Empirical Finance}, vol.~1,
  no.~2, pp. 211--248, 1994. [Online]. Available:
  \url{https://www.sciencedirect.com/science/article/pii/0927539894900043}
\BIBentrySTDinterwordspacing

\bibitem{estoup1916gammes}
J.-B. Estoup, \emph{Gammes st{\'e}nographiques: m{\'e}thode et exercices pour
  l'acquisition de la vitesse}.\hskip 1em plus 0.5em minus 0.4em\relax Institut
  st{\'e}nographique, 1916.

\bibitem{zipf1949human}
G.~K. Zipf, \emph{Human behavior and the principle of least effort.}, ser.
  Human behavior and the principle of least effort.\hskip 1em plus 0.5em minus
  0.4em\relax Oxford, England: Addison-Wesley Press, 1949.

\bibitem{broder2000graph}
\BIBentryALTinterwordspacing
A.~Broder, R.~Kumar, F.~Maghoul, P.~Raghavan, S.~Rajagopalan, R.~Stata,
  A.~Tomkins, and J.~Wiener, ``Graph structure in the web,'' \emph{Computer
  Networks}, vol.~33, no.~1, pp. 309--320, 2000. [Online]. Available:
  \url{https://www.sciencedirect.com/science/article/pii/S1389128600000839}
\BIBentrySTDinterwordspacing

\bibitem{delvin2021hybrid}
P.~Devlin, J.~Kepner, A.~Luo, and E.~Meger, ``Hybrid power-law models of
  network traffic,'' in \emph{2021 IEEE International Parallel and Distributed
  Processing Symposium Workshops (IPDPSW)}, June 2021, pp. 280--287.

\bibitem{huberman1999growth}
\BIBentryALTinterwordspacing
B.~A. Huberman and L.~A. Adamic, ``Growth dynamics of the {W}orld-{W}ide
  {W}eb,'' \emph{Nature}, vol. 401, no. 6749, pp. 131--131, Sep 1999. [Online].
  Available: \url{https://doi.org/10.1038/43604}
\BIBentrySTDinterwordspacing

\bibitem{crovella1998heavy}
M.~E. Crovella, M.~S. Taqqu, and A.~Bestavros, \emph{Heavy-Tailed Probability
  Distributions in the {W}orld {W}ide {W}eb}.\hskip 1em plus 0.5em minus
  0.4em\relax USA: Birkhauser Boston Inc., 1998, p. 3–25.

\bibitem{mahanti2013tale}
A.~Mahanti, N.~Carlsson, A.~Mahanti, M.~Arlitt, and C.~Williamson, ``A tale of
  the tails: Power-laws in internet measurements,'' \emph{IEEE Network},
  vol.~27, no.~1, pp. 59--64, January 2013.

\bibitem{faloutsos1999on}
\BIBentryALTinterwordspacing
M.~Faloutsos, P.~Faloutsos, and C.~Faloutsos, ``On power-law relationships of
  the internet topology,'' \emph{SIGCOMM Comput. Commun. Rev.}, vol.~29, no.~4,
  p. 251–262, aug 1999. [Online]. Available:
  \url{https://doi.org/10.1145/316194.316229}
\BIBentrySTDinterwordspacing

\bibitem{nair2022fundamentals}
J.~Nair, A.~Wierman, and B.~Zwart, \emph{The fundamentals of heavy tails:
  Properties, emergence, and estimation}.\hskip 1em plus 0.5em minus
  0.4em\relax Cambridge University Press, 2022, vol.~53.

\bibitem{poisson1829suite}
S.~D. Poisson, ``Suite du m{\'e}moire sur la probabilit{\'e} des r{\'e}sultats
  moyens des observations, ins{\'e}r{\'e} dans la connaissance des tems de
  l’ann{\'e}e 1827,'' \emph{Connaissance des tems}, pp. 3--22, 1829.

\bibitem{poisson1827probabilite}
------, ``Sur la probabilit{\'e} des r{\'e}sultats moyens des observations,''
  \emph{Conn. des temps pour}, vol. 1824, pp. 273--302, 1827.

\bibitem{demirandacardoso2021graphical}
\BIBentryALTinterwordspacing
J.~V. de~Miranda~Cardoso, J.~Ying, and D.~Palomar, ``Graphical models in
  heavy-tailed markets,'' in \emph{Advances in Neural Information Processing
  Systems}, M.~Ranzato, A.~Beygelzimer, Y.~Dauphin, P.~Liang, and J.~W.
  Vaughan, Eds., vol.~34.\hskip 1em plus 0.5em minus 0.4em\relax Curran
  Associates, Inc., 2021, pp. 19\,989--20\,001. [Online]. Available:
  \url{https://proceedings.neurips.cc/paper\_files/paper/2021/file/a64a034c3cb8
  eac64eb46ea474902797-Paper.pdf}
\BIBentrySTDinterwordspacing

\bibitem{cherapanamjeri2020algorithms}
\BIBentryALTinterwordspacing
Y.~Cherapanamjeri, S.~B. Hopkins, T.~Kathuria, P.~Raghavendra, and
  N.~Tripuraneni, ``Algorithms for heavy-tailed statistics: Regression,
  covariance estimation, and beyond,'' in \emph{Proceedings of the 52nd Annual
  ACM SIGACT Symposium on Theory of Computing}, ser. STOC 2020.\hskip 1em plus
  0.5em minus 0.4em\relax New York, NY, USA: Association for Computing
  Machinery, 2020, p. 601–609. [Online]. Available:
  \url{https://doi.org/10.1145/3357713.3384329}
\BIBentrySTDinterwordspacing

\bibitem{stehlik2010favorable}
\BIBentryALTinterwordspacing
M.~Stehl{\'i}k, R.~Potock{\'y}, H.~Waldl, and Z.~Fabi{\'a}n, ``On the favorable
  estimation for fitting heavy tailed data,'' \emph{Computational Statistics},
  vol.~25, no.~3, pp. 485--503, Sep 2010. [Online]. Available:
  \url{https://doi.org/10.1007/s00180-010-0189-1}
\BIBentrySTDinterwordspacing

\bibitem{benaych-georges2014central}
\BIBentryALTinterwordspacing
F.~Benaych-Georges, A.~Guionnet, and C.~Male, ``Central limit theorems for
  linear statistics of heavy tailed random matrices,'' \emph{Communications in
  Mathematical Physics}, vol. 329, no.~2, pp. 641--686, Jul 2014. [Online].
  Available: \url{https://doi.org/10.1007/s00220-014-1975-3}
\BIBentrySTDinterwordspacing

\bibitem{kepner2022temporal}
J.~Kepner, M.~Jones, D.~Andersen, A.~Buluc, C.~Byun, K.~Claffy, T.~Davis,
  W.~Arcand, J.~Bernays, D.~Bestor, W.~Bergeron, V.~Gadepally, D.~Grant,
  M.~Houle, M.~Hubbell, H.~Jananthan, A.~Klein, C.~Meiners, L.~Milechin,
  A.~Morris, J.~Mullen, S.~Pisharody, A.~Prout, A.~Reuther, A.~Rosa, S.~Samsi,
  D.~Stetson, C.~Yee, and P.~Michaleas, ``Temporal correlation of internet
  observatories and outposts,'' in \emph{2022 IEEE International Parallel and
  Distributed Processing Symposium Workshops (IPDPSW)}, 2022, pp. 247--254.

\bibitem{lumsdaine2007challenges}
A.~Lumsdaine, D.~Gregor, B.~Hendrickson, and J.~Berry, ``Challenges in parallel
  graph processing.'' \emph{Parallel Processing Letters}, vol.~17, pp. 5--20,
  03 2007.

\bibitem{kolda2009tensor}
\BIBentryALTinterwordspacing
T.~G. Kolda and B.~W. Bader, ``Tensor decompositions and applications,''
  \emph{SIAM Review}, vol.~51, no.~3, pp. 455--500, 2009. [Online]. Available:
  \url{https://doi.org/10.1137/07070111X}
\BIBentrySTDinterwordspacing

\bibitem{hilbert2011world}
\BIBentryALTinterwordspacing
M.~Hilbert and P.~López, ``The world's technological capacity to store,
  communicate, and compute information,'' \emph{Science}, vol. 332, no. 6025,
  pp. 60--65, 2011. [Online]. Available:
  \url{https://www.science.org/doi/abs/10.1126/science.1200970}
\BIBentrySTDinterwordspacing

\bibitem{kepner2009parallel}
\BIBentryALTinterwordspacing
J.~Kepner, \emph{Parallel {MATLAB} for Multicore and Multinode
  Computers}.\hskip 1em plus 0.5em minus 0.4em\relax Society for Industrial and
  Applied Mathematics, 2009. [Online]. Available:
  \url{https://epubs.siam.org/doi/abs/10.1137/1.9780898718126}
\BIBentrySTDinterwordspacing

\bibitem{kepner2011graph}
\BIBentryALTinterwordspacing
J.~Kepner and J.~Gilbert, \emph{Graph Algorithms in the Language of Linear
  Algebra}, J.~Kepner and J.~Gilbert, Eds.\hskip 1em plus 0.5em minus
  0.4em\relax Society for Industrial and Applied Mathematics, 2011. [Online].
  Available: \url{https://epubs.siam.org/doi/abs/10.1137/1.9780898719918}
\BIBentrySTDinterwordspacing

\bibitem{kepner2022new}
J.~Kepner, K.~Cho, K.~Claffy, V.~Gadepally, S.~McGuire, L.~Milechin, W.~Arcand,
  D.~Bestor, W.~Bergeron, C.~Byun, M.~Hubbell, M.~Houle, M.~Jones, A.~Prout,
  A.~Reuther, A.~Rosa, S.~Samsi, C.~Yee, and P.~Michaleas, \emph{Massive Graph
  Analytics}, ser. Data Science Series.\hskip 1em plus 0.5em minus 0.4em\relax
  Chapman and Hall, 2022, ch. New Phenomena in Large-Scale Internet Traffic.

\bibitem{gull1988bayesian}
\BIBentryALTinterwordspacing
S.~F. Gull, \emph{Bayesian Inductive Inference and Maximum Entropy}.\hskip 1em
  plus 0.5em minus 0.4em\relax Dordrecht: Springer Netherlands, 1988, pp.
  53--74. [Online]. Available:
  \url{https://doi.org/10.1007/978-94-009-3049-0\_4}
\BIBentrySTDinterwordspacing

\end{thebibliography}

\section*{Appendix A: Generalized Cauchy Distribution}

\label{derivation of generalized cauchy distribution appendix}

Let---as in the derivation of the Cauchy distribution---$X$ and $\Theta$ be random variables for the first coordinate of the coastline position and the azimuth, respectively. This implies
\begin{align*}
	p_X(x) & = p_\Theta(\theta(x)) \left|\frac{d\theta(x)}{dx}\right| \tag*{} \\
	& = \frac{1}{\beta - \alpha} \cdot \frac{(y_0 - f(x)) + x f'(x)}{x^2 + (f(x) - y_0)^2}, \label{generalized lighthouse distribution}
\end{align*}
when the coastline is the image of a real-valued function of a real variable $f$.  Given the lighthouse is at position $(0, y_0)$, if for each azimuth $\theta \in (\alpha, \beta)$ the flash is seen on the coastline at a distance of $r(\theta)$ away from the lighthouse's position, then we have
\begin{align*}
	x & = r(\theta) \sin(\theta), \\
	y & = y_0 - r(\theta) \cos(\theta)
\end{align*}
As such, $x^2 + (y_0 - y)^2 = r(\theta)^2$, hence $r(\theta) = \sqrt{x^2 + (y_0 - y)^2}$ since $r(\theta) \geq 0$ for all values of $\theta$. With this in mind, we may solve for $\theta$ to find
\begin{equation*}
	\theta = \arccos\left(\frac{y_0 - y}{\sqrt{(y_0 - y)^2 + x^2}}\right)
\end{equation*}
For brevity, let $\hat{y} \coloneq y_0 - y$, noting that $\hat{y}' = \frac{d\hat{y}}{dx} = -\frac{dy}{dx}$. Then
\begin{align*}
	\frac{d\theta}{dx} & = -\frac{\frac{\hat{y}' \sqrt{x^2 + \hat{y}^2} - \hat{y} \frac{2 x x' + 2 \hat{y} \hat{y}'}{2\sqrt{x^2 + \hat{y}^2}}}{x^2 + \hat{y}^2}}{\sqrt{1 - \left(\frac{\hat{y}}{\sqrt{x^2 + \hat{y}^2}}\right)^2}} \\
	& = -\frac{\frac{\hat{y}' (x^2 + \hat{y}^2) - \hat{y} (x + \hat{y} \hat{y}')}{(x^2 + \hat{y}^2)^{3/2}}}{\sqrt{\frac{x^2 + \hat{y}^2}{x^2 + \hat{y}^2} - \frac{\hat{y}^2}{x^2 + \hat{y}^2}}} \\
	& = -\frac{x^2 \hat{y}' + \hat{y}^2 \hat{y}' - x \hat{y} - \hat{y}^2 \hat{y}'}{(x^2 + \hat{y}^2)^{3/2} \sqrt{\frac{x^2}{x^2 + \hat{y}^2}}} \\
	& = -\frac{x}{|x|} \frac{x \hat{y}' - \hat{y}}{x^2 + \hat{y}^2} \\
%	& = \sgn(x) \frac{\hat{y} - x \hat{y}'}{x^2 + \hat{y}^2} \\
	& = \sgn(x) \frac{(y_0 - y) + x y'}{x^2 + (y - y_0)^2}
\end{align*}

\section*{Appendix B: Limitations of Sine-Squared Heuristic}
\label{limitations of sine-square heuristic appendix}

As a mechanism for testing the applicability of our heuristic \ref{geometric heuristic}, Figure~\ref{original pdfs versus regenerated pdfs figure} (top) plots the probability density functions of the Gaussian distribution ($\mathrm{Gaussian}(0, 1)$), the Cauchy distribution ($\mathrm{Cauchy}(1)$), the modified Cauchy distribution ($\mathrm{ModCauchy}(1, 3/4)$), and the generalized Cauchy distributions obtained by using as coastline the lower unit semicircle centered at $(0, 1)$, the $45^\circ$ line $y = x$, and the $135^\circ$ line $y = -x$. In addition, the generalized Cauchy distributions obtained by using the lighthouse position $(0, 1)$ and coastlines generated from the sine-squared heuristic \ref{geometric heuristic} are plotted against the original distributions. The values of $\beta$ used within the sine-squared heuristic have been chosen between $0$ and $5$ to minimize the L2 norm between the original and re-obtained probability density functions. We observed that the heuristic works best when the probability density function $p$ is decreasing, so for the generalized Cauchy distribution corresponding to the lower unit semicircle centered at $(0, 1)$ the heuristic was applied on the interval $[-0.9, 0]$ and flipped horizontally across the $y$-axis. A similar approach was needed for the generalized Cauchy distributions corresponding to the coastlines $y = x$ and $y = -x$ by separating the intervals $[-5, 1/2]$, $[1/2, 5]$ for the former ($(1/2, 1/2)$ is the closest point on $y = x$ to $(0, 1)$) and $[-5, -1/2]$, $[-1/2, 5]$ for the latter (similarly, $(-1/2, 1/2)$ is the closest point on $y = -x$ to $(0, 1)$), though it was necessary to explicitly set the value of $f_1$ prior to the recursion in Algorithm~\ref{heuristic algorithm} since otherwise the algorithm sets $f_1 = 0$ which would lead to a discontinuous coastline.

For appropriately selected $\beta$, the heuristic qualitatively re-obtained the probability density functions in all cases, and obtained very closely matched quantitative results in the cases of $\mathrm{Gaussian}(0, 1)$, $\mathrm{Cauchy}(1)$, $\mathrm{GenCauchy}(0, 1; y=x)$, and $\mathrm{GenCauchy}(0, 1; y=-x)$. The quantitative differences in the heuristic re-obtained probability density function for $\mathrm{ModCauchy}(1, 3/4)$ are sufficiently small as to not effective the overall geometric interpretation of the results in Figure~\ref{caida coastlines figure}. The decay rates were roughly proportional and might be attributable to normalization difficulties since $\mathrm{ModCauchy}(1, 3/4)$ is not supported on all of $\mathbb{R}$ and hence normalization constants depend upon the arbitrarily declared domain. The larger differences seen in the $\mathrm{GenCauchy}(0, 1; \text{semicircle})$ probability density function are perhaps due to its high positive concavity and warrant further investigation. 

\begin{figure}[htbp]
	\begin{center}
		\includegraphics[width=1.0\columnwidth]{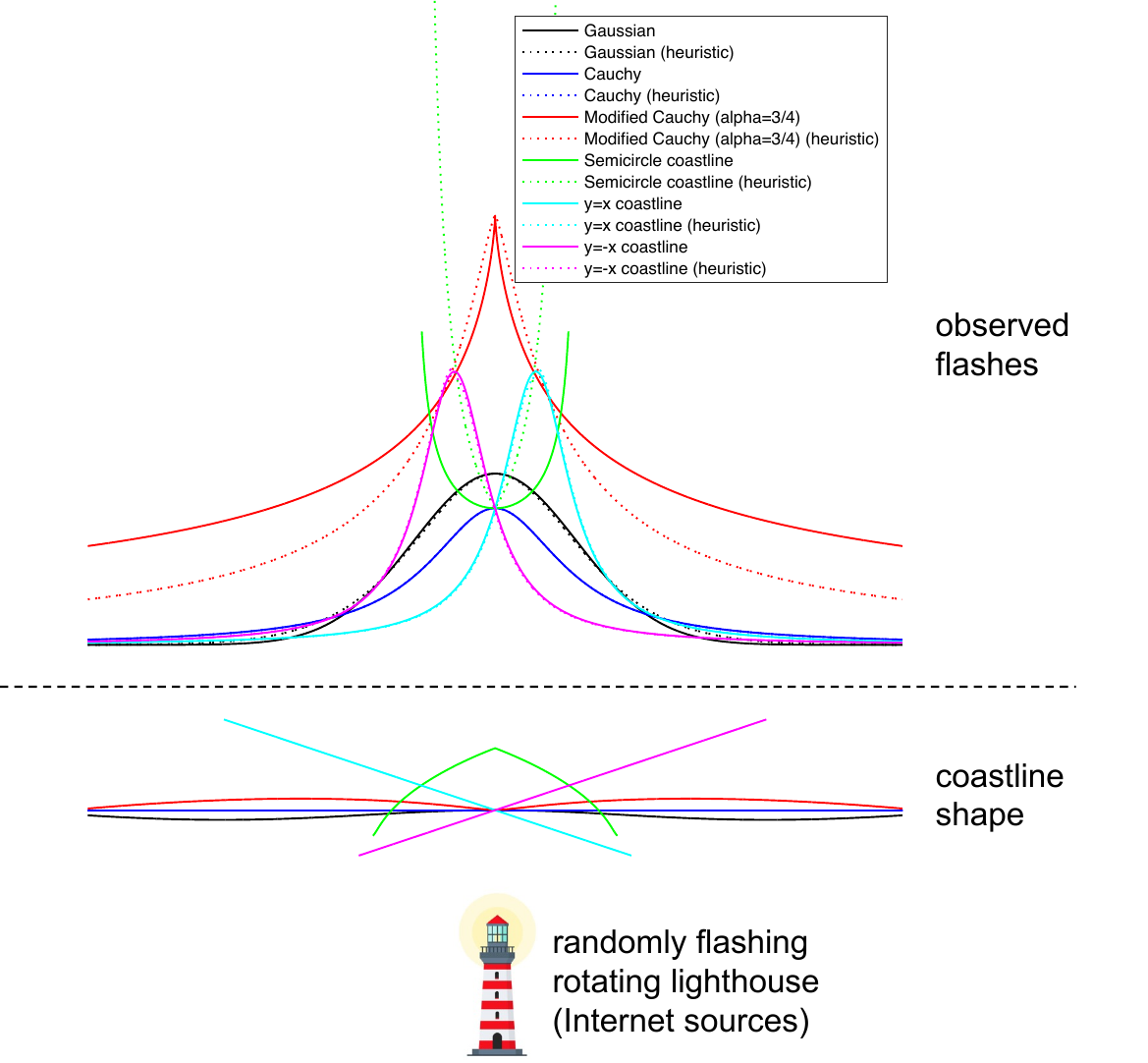}
	\end{center}
	\caption{ (top) Plots of distributions $\mathrm{Gaussian}(0, 1)$, $\mathrm{Cauchy}(1)$, $\mathrm{ModCauchy}(1, 3/4)$, $\mathrm{GenCauchy}(0, 1; \text{semicircle})$, $\mathrm{GenCauchy}(0, 1; y=x)$, and $\mathrm{GenCauchy}(0, 1; y=-x)$ (solid plots) versus the generalized Cauchy distributions obtained by using the lighthouse position $(0, 1)$ and coastlines generated from the sine-squared heuristic \ref{geometric heuristic} (dashed plots with same color as the corresponding original distribution). The values of the proportionality constant $\beta$ utilized in the heuristic \ref{geometric heuristic} were chosen beween $0$ and $5$ to minimize L2 norm after normalizing the generated generalized Cauchy distributions to agree with their corresponding original distributions at the minimum $x$ value. (bottom) Plots of the coastlines generated from the solid-line probability distributions and generating the dashed-line probability distributions illustrated in Figure~\ref{original pdfs versus regenerated pdfs figure}, matched by color.} \label{heuristic coastlines figure.} \label{original pdfs versus regenerated pdfs figure}
\end{figure}

\end{document}